\documentclass[12pt]{article}

\usepackage[hmargin=1.3in,vmargin=1.3in]{geometry}
\usepackage{bbding}
\usepackage{mathrsfs}
\usepackage{}
\usepackage{amsfonts}
\usepackage{multirow}
\usepackage{amsfonts,amssymb,amsmath,amsthm,bm}

\newtheorem{theorem}{Theorem}
\newtheorem{example}{Example}

\newtheorem{remark}{Remark}

\newtheorem{lemma}{Lemma}
\newtheorem{proposition}{Proposition}

\begin{document}

\title{Optimal cyclic $(r,\delta)$ locally repairable codes with unbounded length\thanks{This paper was presented in part at the 2018 IEEE Information Theory Workshop (ITW).}}
\author{\small  Weijun Fang\thanks{Corresponding Author} $^{,\natural}$  \   \ Fang-Wei Fu$^{\natural}$ \\
\small $^\natural$  Chern Institute of Mathematics and LPMC, Nankai University, Tianjin, China\\
\small Emails: nankaifwj@163.com, fwfu@nankai.edu.cn\\
}
\date{}
\maketitle
\thispagestyle{empty}
\begin{abstract}
Locally repairable codes with locality $r$ ($r$-LRCs for short) were introduced by Gopalan et al.  \cite{1} to recover a failed node of the code from at most other $r$  available nodes. And then $(r,\delta)$ locally repairable codes ($(r,\delta)$-LRCs for short) were produced by Prakash et al. \cite{2} for tolerating multiple failed nodes. An $r$-LRC can be viewed as an $(r,2)$-LRC. An $(r,\delta)$-LRC is called optimal if it achieves the Singleton-type bound. It has been a great challenge to construct $q$-ary optimal $(r,\delta)$-LRCs with length much larger than $q$. Surprisingly, Luo et al. \cite{3} presented a construction of $q$-ary optimal $r$-LRCs of minimum distances 3 and 4 with unbounded
lengths (i.e., lengths of these codes are independent of $q$) via cyclic codes.

In this paper, inspired by the work of \cite{3}, we firstly construct two classes of optimal cyclic $(r,\delta)$-LRCs with unbounded
lengths and  minimum distances $\delta+1$ or $\delta+2$, which generalize the results about the $\delta=2$ case given in \cite{3}. Secondly, with a slightly stronger condition, we present a construction of optimal cyclic $(r,\delta)$-LRCs with unbounded
length and larger minimum distance $2\delta$. Furthermore, when $\delta=3$, we give another class of optimal cyclic $(r,3)$-LRCs with unbounded
length and minimum distance $6$.
\end{abstract}


\maketitle

%-------------------------------------------------------------------------------
\section{Introduction}

Motivated by applications to distributed storage, locally repairable codes (LRC) were introduced by Gopalan et al.  \cite{1}, which have attracted great
attention of researchers recently. Such repair-efficient codes are already used in the Hadoop
Distributed File System RAID  by Facebook and Windows
Azure Storage \cite{18,19}. Let $C$ be an $[n,k,d]$ linear code. The $i$-th code symbol of $C$
is said to have locality $r$ ($1 \leq r \leq k$) if it can be recovered
by accessing at most $r$ other symbols in $C$, i.e., the $i$-th code symbol can be expressed as a linear combination of $r$ other symbols.  If all
symbols of $C$ have locality $r$, then $C$ is called an $r$-LRC. Any $r$-LRC has to satisfy the Singleton-type bound, which was proposed in \cite{1}:
\begin{equation}\label{1}
  d \leq n - k - \lceil\frac{k}{r}\rceil +2.
\end{equation}

When  multiple
node failures occur in a distributed storage system, the local recovery process for a failed node may not proceed
successfully. In order to overcome this problem, Prakash et al. \cite{2} introduced the
concept of $(r,\delta)$-locality, which generalize the $r$-locality. The $i$-th code symbol of $C$
is said to have locality $(r,\delta)$ ($1 \leq r \leq k$ and $\delta \geq 2$), if there exists a subset $R_{i} \subseteq \{1, 2, \ldots, n \}$ such that $i \in R_{i}, |R_{i}| \leq r+\delta-1$ and the punctured code $C|_{R_{i}}$ has minimum distance $d(C|_{R_{i}}) \geq \delta$.  And $C$ is called an $(r, \delta)$-LRC if all
nodes of $C$ have locality $(r, \delta)$. When $\delta=2$, it is easy to see that an $(r, \delta)$-LRC degenerates to an $r$-LRC. For an $[n, k, d]$ linear code with $(r, \delta)$-locality, Prakash et al. \cite{2} gave the
following Singleton-type bound:
\begin{equation}\label{2}
  d \leq n - k + 1- (\lceil\frac{k}{r}\rceil-1)(\delta-1).
\end{equation}
An $[n, k, d]$-LRC with locality $(r,\delta)$ (resp. $r$) is called optimal if it achieves the bound (2) (resp. (1)).

Lots of works have been proposed for construction of optimal LRCs (\cite{2,3,4,5,6,7,8,9,10,11,12,13,14}). In \cite{11}, optimal  $(r, \delta)$-LRCs were constructed with alphabet size which is exponential
in code length $n$. A class of optimal $(r,\delta)$-LRCs with length $n=\lceil\frac{k}{r}\rceil(r+\delta-1)$ was obtained in \cite{2} for $n < q$. A breakthrough construction given in \cite{4} produces a family of optimal $(r, \delta)$-LRCs via subcodes of Reed-Solomon codes. The length of these codes can go up to the alphabet size. By employing the techniques of cyclic MDS codes \cite{16}, Chen et al. \cite{5} obtained several classes of $q$-ary optimal cyclic $(r, \delta)$-LRCs with length $n \mid (q+1)$. In \cite{6}, Jin et al. also constructed a family of $q$-ary optimal $r$-LRCs with length up to $q+1$ by using the automorphism group
of rational function fields. By studying the algebraic structures of elliptic curves, Ma et al. \cite{7} construct a family of $q$-ary optimal $r$-LRCs of length up to $q+2\sqrt{q}$. One natural question is that how long can an optimal LRC be? Surprisingly, it was shown in \cite{3} that there exist optimal cyclic $r$-LRCs with unbounded lengths and minimum distances 3 or 4.

In this paper, we generalize the work of \cite{3} to the $(r,\delta)$-LRCs for general $\delta \geq 2$. Firstly, we construct two classes of optimal cyclic $(r,\delta)$-LRCs with unbounded
length and  minimum distances $\delta+1$ and $\delta+2$.  The main results of \cite{3} then can be seen as the $\delta=2$ case of ours. Secondly, under a slightly stronger condition, we present a construction of optimal cyclic $(r,\delta)$-LRCs with unbounded
length and larger minimum distance $2\delta$. When $\delta=3$, with a modification of this construction, we construct another class of optimal cyclic $(r,3)$-LRCs with unbounded length and minimum distance 6. More precisely, we have the following
main results in this paper.
\begin{itemize}
  \item (i) Suppose $\gcd(q, n)=1$. Let $r, \delta \geq 2$ with $(r+\delta-1) \mid \gcd(n, q-1)$, then there exists a $q$-ary optimal cyclic $(r, \delta)$-LRC with length $n$ and minimum distance $\delta+1$.
  \item  (ii) Suppose $\gcd(q, n)=1$. Let $r \geq 3, \delta \geq 2$ with $(r+\delta-1) \mid \gcd(n, q-1)$, and  $\gcd(\frac{n}{r+\delta-1}, r+\delta-1) \mid \delta$, then there exists a $q$-ary optimal cyclic $(r, \delta)$-LRC with length $n$ and minimum distance $\delta+2$.
  \item (iii) Suppose $\gcd(q, n)=1$. Let $r \geq \delta+1$ with $(r+\delta-1) \mid \gcd(n, q-1)$, and  $\gcd(\frac{n}{r+\delta-1}, r+\delta-1) =1$, then there exists a $q$-ary optimal cyclic $(r, \delta)$-LRC with length $n$ and minimum distance $2\delta$.
  \item (iv) Suppose $n$ is odd and $\gcd(q, n)=1$. Let $r \geq 4$ with $(r+2) \mid \gcd(n, q+1)$, and  $\gcd(\frac{n}{r+2}, r+2) =1$, then there exists a $q$-ary optimal cyclic $(r, 3)$-LRC with length $n$ and minimum distance $6$.
\end{itemize}

The rest of this paper is organized as follows. In Section 2, we review some preliminaries on cyclic codes and present some basic results of cyclic $(r, \delta)$-LRCs. In Section 3, we present our constructions of optimal cyclic $(r, \delta)$-LRCs. We use some conclusions to end this paper in Section 4.

\section{Preliminaries}

In this section, we review some preliminaries on cyclic codes and present some basic results of cyclic $(r, \delta)$-LRCs.

\subsection{Cyclic codes}
Throughout this paper, we let $q$ be a prime power and $\mathbb{F}_{q}$ be a finite field with size $q$. A linear code $C$ of length $n$ over $\mathbb{F}_{q}$ is called cyclic if $(c_{0}, c_{1}, \ldots, c_{n - 1}) \in C$ implies that $( c_{n - 1}, c_{0}, \ldots, c_{n-2}) \in C$. It is well-known that a $q$-ary cyclic code $C$ of length $n$ can be identified with an ideal of
the ring  $\mathbb{F}_{q}[x]/(x^{n} - 1)$, where
$\gcd(n, q) = 1$. Since $\mathbb{F}_{q}[x]/(x^{n} - 1)$ is principal, every ideal of $\mathbb{F}_{q}[x]/(x^{n} - 1)$ is generated by a monic polynomial $g(x)$ with $g(x) \mid (x^{n} - 1)$.  $g(x)$ is called the generator polynomial of $C$.

Let $s$ be the order of $q$ modulo $n$, i.e., the least number of $i$ such that $n \mid (q^{i}-1)$. Then $\mathbb{F}_{q^{s}}$ is a splitting field of $x^{n}-1$. Let $\xi \in \mathbb{F}_{q^{s}}$ be a primitive $n$-th root of unity. Let $C$ be
a $q$-ary $[n, k, d]$-cyclic code with generator polynomial $g(x)$. The
zeros set $Z = \{ \xi^{i_{j}} \mid g(\xi^{i_{j}}) = 0,j = 1, 2, \ldots, n-k \}$ of $g(x)$
is called the complete defining set of $C$. The following lemma is a simple generalization of the well-known BCH Bound.

\begin{lemma}(Generalized BCH Bound, \cite{15})
Let $C$ be a $q$-ary cyclic code of length $n$, where $\gcd(n,q)=1$. Let $g(x)$ be the generator
polynomial of $C$, and $\xi$ be a primitive $n$-th root of unity. If
$g(x$) has $\xi^{u},\xi^{u+b}, \ldots, \xi^{u+(d-2)b}$ among its zeros, where $u$ is an integer and $\gcd(b,n)=1$. Then the minimum
distance of $C$ is at least $d$.
\end{lemma}
\subsection{Cyclic $(r,\delta)$-LRCs}
When $(r+\delta-1) \nmid n$, it was proved in \cite[Theorem 10]{10} that there is no $(r, \delta)$-LRCs with $r \mid k$ achieving the bound (2). Thus, throughout this paper,
we assume that $(r + \delta- 1) \mid n$ and $\gcd(n,q)=1$. Let $\xi \in \mathbb{F}_{q^{s}}$ be
a primitive $n$-th root of unity, where $\mathbb{F}_{q^{s}}$ is the   splitting field of $x^{n}-1$.

When $n \mid (q-1)$, Tamo et al. provided a useful condition to ensure a cyclic code
has locality $r$ in \cite[Proposition 3.4]{10}. B. Chen et al. \cite{5} then generalized
their results to the cyclic $(r, \delta)$-LRCs of length $n \mid (q-1)$ or $(q+1)$. Actually, their results can be easily generalized for general $n$ with $\gcd(n,q)=1$, which are presented as follows.

\begin{lemma}
Suppose that $\gcd(n,q)=1$, $(r+\delta-1) \mid n$ and $\rho=\frac{n}{r+\delta-1}$. Let $\ell_{1} < \ell_{2} < \cdots < \ell_{\delta- 1}$
be an arithmetic progression with $\delta - 1$ items and common
difference $b$, where $\gcd(b, n) = 1$. Consider a $(\delta-1)\rho \times n$ matrix
$H$ with the rows
\[h_{j}^{(\ell_{i})}=(1,\xi^{j(r + \delta- 1)+\ell_{i}},\xi^{2(j(r + \delta- 1)+\ell_{i})},\ldots,\xi^{(n-1)(j(r + \delta- 1)+\ell_{i})}),\]
where $i=1, 2, \ldots, \delta-1$, $j=1, 2, \ldots, \rho$. Then
all the cyclic shifts of the row vectors of weight $r + \delta-1$ in
the following $(\delta- 1) \times n$-matrix
\[V=\left(
    \begin{array}{ccccccc}
      10\cdots0 & \xi^{\rho \ell_{1}} & 0\cdots0 & (\xi^{\rho \ell_{1}})^{2} & 0\cdots0 & (\xi^{\rho \ell_{1}})^{r+\delta-2} & 0\cdots0 \\
      10\cdots0 & \xi^{\rho \ell_{2}} & 0\cdots0 & (\xi^{\rho \ell_{2}})^{2} & 0\cdots0 & (\xi^{\rho \ell_{2}})^{r+\delta-2} & 0\cdots0 \\
      \vdots & \vdots & \vdots & \vdots & \ddots & \vdots & \vdots \\
       10\cdots0 & \xi^{\rho \ell_{\delta-1}} & 0\cdots0 & (\xi^{\rho \ell_{\delta-1}})^{2} & 0\cdots0 & (\xi^{\rho \ell_{\delta-1}})^{r+\delta-2} & 0\cdots0 \\
    \end{array}
  \right)\]
are contained in the row space of $H$ over $\mathbb{F}_{q^{s}}$.
\end{lemma}

\begin{proof}
Note that for any $i=1, 2, \ldots, \delta$ and $m=0, 2, \ldots, n-1$,
\[\frac{1}{\rho}\sum_{j=1}^{\rho}\xi^{m(j(r+\delta-1)+\ell_{i})}=\left\{\begin{array}{cc}
                                                                       0 & \textnormal{ if }\rho \nmid m; \\
                                                                       \xi^{m\ell_{i}} & \textnormal{ if } \rho \mid m.
                                                                     \end{array}\right.
\]
Thus the row vectors of $V$ are contained in the row space of $H$ over $\mathbb{F}_{q^{s}}$. It is easy to see that the
row space of $H$ over $\mathbb{F}_{q^{s}}$ is closed under cyclic shifts, thus the lemma follows.
\end{proof}

\begin{proposition}
Suppose that $\gcd(n,q)=1$, $(r+\delta-1) \mid n$ and $\rho=\frac{n}{r+\delta-1}$. Let $C$ be a cyclic code of length $n$ over $\mathbb{F}_{q}$ with complete defining set $Z$. Let $\ell_{1} < \ell_{2} < \cdots < \ell_{\delta- 1}$
be an arithmetic progression with $\delta - 1$ items and common
difference $b$, where $\gcd(b, n) = 1$.  If $Z$ contains some cosets of the group of $\rho$-th
roots of unity $\bigcup_{\ell}L_{\ell}$, where
\[L_{\ell}=\{\xi^{i} \mid i\textnormal{ mod }(r+\delta -1) =\ell \}, \ell=\ell_{1}, \ell_{2}, \ldots, \ell_{\delta-1},\]
then $C$ has $(r, \delta)$-locality.
\end{proposition}

\begin{proof}
Note that
\[L_{\ell}=\{\xi^{j(r + \delta- 1)+\ell}\}_{j=1}^{\rho}.\] Let $H$ and $V$ be the matrices defined in Lemma 2, then $H$ forms a parity-check matrix of the cyclic code $C$. By Lemma 2, $V$ is contained in the row space of the parity check matrix $H$.  Let $V'$ be the non-zero columns of $V$, i.e.,
\[V'=\left(
  \begin{array}{ccccc}
    1 & \xi^{\rho \ell_{1}} &(\xi^{\rho \ell_{1}})^{2} & \cdots & (\xi^{\rho \ell_{1}})^{r+\delta-2} \\
   1 & \xi^{\rho \ell_{2}} &(\xi^{\rho \ell_{2}})^{2} & \cdots & (\xi^{\rho \ell_{2}})^{r+\delta-2} \\
    \vdots & \vdots & \vdots &\ddots & \vdots \\
    1 & \xi^{\rho \ell_{\delta-1}} &(\xi^{\rho \ell_{\delta-1}})^{2} & \cdots & (\xi^{\rho \ell_{\delta-1}})^{r+\delta-2} \\
  \end{array}
\right).\]
Since $(r+\delta-1) \mid n$ and $\gcd(b, n) = 1$, $\gcd(b, r+\delta-1) = 1$. Thus for any $1 \leq i \neq j \leq \delta-1$, $\ell_{j}-\ell_{i}=b(j-i)$ is not divisible by $r+\delta-1$. Hence $\xi^{\rho\ell_{i}} \neq \xi^{\rho\ell_{j}}$ and $V'$ is a parity-check matrix of an $[r+\delta-1, r, \delta]$ Reed-Solomon code. Then we can obtain that $C$ has $(r, \delta)$-locality similarly as \cite[Proposition 6]{5}.
\end{proof}

\section{Constructions}
In this section, by generalizing the technique proposed in \cite{3}, we provide four classes of $q$-ary optimal $(r, \delta)$-LRCs via cyclic codes. The lengths of these codes are unbounded, i.e., lengths are independent of $q$.

\begin{theorem}
Let $q$ be a prime power and $n$ be positive integer with $\gcd(n, q)=1$. Let $r, \delta \geq 2$ such that $(r+\delta -1) \mid \gcd(n, q-1)$. Then there exists a $q$-ary optimal cyclic $(r, \delta)$-LRC of length $n$ and minimum distance $\delta+1$.
\end{theorem}

\begin{proof}
  Let $\xi \in \mathbb{F}_{q^{s}}$ be
a primitive $n$-th root of unity and $\alpha = \xi^{\rho}$, where $\rho=\frac{n}{r+\delta-1}$.  Since $(r+\delta-1) \mid (q-1)$, $\alpha^{q-1}=(\xi^{\rho})^{q-1}=(\xi^{n})^{\frac{q-1}{r+\delta-1}}=1$, thus $\alpha \in \mathbb{F}_{q}$. Let
\[g(x)=(x-1)(x^{\rho}-\alpha)(x^{\rho}-\alpha^{2})\cdots(x^{\rho}-\alpha^{\delta-1}).\]
Then $g(x)$ is a polynomial over $\mathbb{F}_{q}$ and $g(x) \mid (x^{n}-1)$ since all roots of $g(x)$ are $n$-th roots of unity and they are distinct. Let $C$ be the cyclic code with generator polynomial $g(x)$. Then the dimension of $C$ is $k=n-\deg(g(x))=n-(\rho(\delta-1)+1)=r\rho-1$. Note that the set of the roots of $g(x)$ contains $\bigcup_{\ell=1}^{\delta-1}L_{\ell},$
where
 \[L_{\ell}=\{\xi^{i} \mid i\textnormal{ mod }(r+\delta -1) =\ell \}.\]
By Proposition 1, $C$ has $(r, \delta)$-locality. As $1, \xi, \ldots, \xi^{\delta-1}$ are roots of $g(x)$, the minimum distance $d$ of $C$ is at least $\delta+1$ by Lemma 1. Note that $\lceil\frac{k}{r}\rceil=\lceil\frac{r\rho-1}{r}\rceil=\rho$. By the bound
(2),
\[d \leq n-(r\rho-1)-(\rho-1)(\delta-1)+1=\delta+1.\] Thus $d=\delta+1$. The proof is completed.
\end{proof}

\begin{remark}
\cite[Theorem 1 (1)]{3} can be seen as the $\delta=2$ case of Theorem 1.
\end{remark}

\begin{example}
Let $r=\delta=3$ and $q=11$,  then by Theorem 1, for any $n=5n'$ with $\gcd(n', 11)=1$, there exists an $11$-ary optimal cyclic $(3, 3)$-LRC of length $n$ and minimum distance $4$.
\end{example}

\begin{theorem}
Let $q$ be a prime power and $n$ be a positive integer with $\gcd(n, q)=1$. Let $r \geq 3$ and $\delta \geq 2$ such that $(r+\delta -1) \mid \gcd(n, q-1)$ and $\gcd(\frac{n}{r+\delta -1}, r+\delta -1) \mid \delta$. Then there exists a $q$-ary optimal cyclic $(r, \delta)$-LRC of length $n$ and minimum distance $\delta+2$.
\end{theorem}

\begin{proof}
  Let $\xi \in \mathbb{F}_{q^{s}}$ be
a primitive $n$-th root of unity and $\alpha = \xi^{\rho}$ be a primitive $(r+\delta-1)$-th root of unity, where $\rho=\frac{n}{r+\delta-1}$.  Since $(r+\delta-1) \mid (q-1)$, $\alpha^{q-1}=(\xi^{\rho})^{q-1}=(\xi^{n})^{\frac{q-1}{r+\delta-1}}=1$, thus $\alpha \in \mathbb{F}_{q}$. Since $\gcd(\rho, r+\delta -1) \mid \delta$, there exist integers $a, b$, such that $a\rho+b(r+\delta -1)=\delta$. Let $\gamma = \alpha^{a} \in \mathbb{F}_{q}$. Then
\[\gamma^{\rho}=\alpha^{a\rho}=\alpha^{\delta-b(r+\delta -1)}=\alpha^{\delta}.\]

Let
\[g(x)=(x-1)(x-\gamma)(x^{\rho}-\alpha)(x^{\rho}-\alpha^{2})\cdots(x^{\rho}-\alpha^{\delta-1}).\]
Then $g(x)$ is a polynomial over $\mathbb{F}_{q}$ and $g(x) \mid (x^{n}-1)$ since all roots of $g(x)$ are $n$-th roots of unity and they are distinct. Let $C$ be the cyclic code with generator polynomial $g(x)$. Then the dimension of $C$ is $k=n-\deg(g(x))=n-(\rho(\delta-1)+2)=r\rho-2$. Note that the set of the roots of $g(x)$ contains $\bigcup_{\ell=1}^{\delta-1}L_{\ell},$
where
 \[L_{\ell}=\{\xi^{i} | i\textnormal{ mod }(r+\delta -1) =\ell \}.\]
By Proposition 1, $C$ has $(r, \delta)$-locality.  Note that $\lceil\frac{k}{r}\rceil=\lceil\frac{r\rho-2}{r}\rceil=\rho$ since $r \geq 3$. By the bound (2),
\[d \leq n-(r\rho-2)-(\rho-1)(\delta-1)+1=\delta+2.\]
Thus to prove $C$ is optimal, we only need to show that $d \geq \delta+2$. By contradiction, we suppose $d \leq \delta+1$. Then there exists a nonzero polynomial $c(x)=\sum_{i=0}^{\delta}c_{i}x^{k_{i}}$ with $0=k_{0} < k_{1}< \cdots < k_{\delta} <n$, such that $g(x) \mid c(x)$. Since $\xi$ is a primitive $n$-th root of unity,
\[\xi^{i} \neq \xi^{j},\textnormal{ for any }0 \leq i \neq j \leq \delta .\]
Let
\[\bm{1}=(1,1,\ldots,1),\]
\[\bm{\gamma}=\left(1, \gamma^{k_{1}},\ldots, \gamma^{k_{\delta}}\right), \]
and
\[\bm{u}_{i,j}=\left(1,\xi^{(i+j(r+\delta-1))k_{1}},\ldots ,\xi^{(i+j(r+\delta-1))k_{\delta}}\right), i=1, 2, \ldots, \delta-1, j=0, 1, \ldots, \rho-1,\]
be the vectors in $\mathbb{F}^{\delta+1}_{q^{s}}$. Since $1, \gamma$ and $\xi^{i+j(r+\delta-1)}$ are roots of $c(x)$,
\begin{equation}\label{3}
  <\bm{1}, \bm{c}>=<\bm{\gamma}, \bm{c}>=<\bm{u}_{i,j}, \bm{c}>=0,
\end{equation}
for all $i=1, 2, \ldots, \delta-1$ and $j=0, 1, \ldots, \rho-1$, where $\bm{c}=(c_{0}, c_{1}, \ldots, c_{\delta})$ and $<,>$ is the canonical Euclidean inner product of $\mathbb{F}^{\delta+1}_{q^{s}}$.
Let $U$ be the vector space spanned by  the vectors $\bm{1}, \bm{\gamma}, \bm{u}_{i,j}$ ($i=1, 2, \ldots, \delta-1$ and $j=0, 1, \ldots, \rho-1$) over $\mathbb{F}_{q^{s}}$. Since $\bm{c} \neq \bm{0}$ and Eq. (3), we have $\dim(U) \leq \delta$, i.e., any $\delta + 1$ vectors of $U$ are linearly dependent. In other words, we have
\begin{center}
\textbf{Fact 1}: Let $M$ be a $(\delta+1)\times(\delta+1)$ matrix whose row vectors belong to $U$, then $\det(M)=0$.
\end{center}

Now, assume that among these $\delta$ integers $k_{i}$, $t$ integers are divisible by $\rho$ and the rest $\delta-t$ integers are not divisible by $\rho$. Then $0 \leq t \leq \delta$. Without loss of generality, we assume that $\rho \mid k_{1}, k_{2}, \ldots, k_{t}$. Then for any $1 \leq i \leq \delta-1$,  we have
\begin{equation}\label{4}
 \sum_{j=0}^{\rho-1}\xi^{(i+j(r+\delta-1))k_{\ell}}=\left\{
\begin{aligned}
   \rho\xi^{ik_{\ell}}, & \textnormal{ if }1 \leq \ell \leq t, \\
 0, & \textnormal{ if } t+1 \leq j \leq \delta.
\end{aligned}\right.
\end{equation}
Set
\[A=\left(
      \begin{array}{c}
        \bm{1} \\
        \bm{u}_{1,0}\\
        \bm{u}_{2,0} \\
        \vdots \\
        \bm{u}_{\delta-1,0} \\
      \end{array}
    \right).\]

\textbf{Case (i)}: $ 0 \leq t \leq \delta-2$. By Eq. (4), we have
\[\bm{u}_{i} \triangleq \bm{u}_{i,0}+\bm{u}_{i,1}+ \cdots +\bm{u}_{i,\rho-1}=\rho(1, \xi^{ik_{1}},\ldots,\xi^{ik_{t}},0,\ldots,0), \]
for $i=1, 2, \ldots, \delta-1$. Each $\bm{u}_{i} \in U$.
Let
\[B=\frac{1}{\rho}\left(
      \begin{array}{c}
        \bm{u}_{1}\\
        \bm{u}_{2} \\
         \vdots \\
        \bm{u}_{t+1}\\
      \end{array}
    \right)=\left(
              \begin{array}{ccccccc}
                1 & \xi^{k_{1}} & \ldots & \xi^{k_{t}} & 0 & \ldots & 0 \\
                1 & \xi^{2k_{1}} & \ldots & \xi^{2k_{t}} & 0 & \ldots & 0 \\
                \vdots & \vdots & \ddots & \vdots & \vdots & \ddots & \vdots \\
                1 & \xi^{(t+1)k_{1}} & \ldots & \xi^{(t+1)k_{t}} & 0 & \ldots & 0 \\
              \end{array}
            \right).\]
    Then the first $t+1$ columns of $B$ form a Vandermonde matrix which is invertible since $\xi^{k_{i}}  \neq \xi^{k_{j}}$ for any $1 \leq i \neq j \leq \delta$. We deduce that $\bm{e}=(1,0,\ldots,0) \in U$. Let
    \[ M=\left(
           \begin{array}{c}
             A \\
             \bm{e} \\
           \end{array}
         \right).\]
         Then
         \begin{eqnarray*}
         % \nonumber to remove numbering (before each equation)
          \det(M) &=& \det\left(
                         \begin{array}{cccc}
                           1 & 1 & \cdots & 1 \\
                           1 & \xi^{k_{1}} & \cdots & \xi^{k_{\delta}} \\
                           \vdots & \vdots & \ddots & \vdots \\
                           1 & \xi^{(\delta-1)k_{1}} & \cdots & \xi^{(\delta-1)k_{\delta}} \\
                           1 & 0 & \cdots & 0 \\
                         \end{array}
                       \right) \\
               &=&\det\left(
                         \begin{array}{cccc}
                          1 &   1 & \cdots & 1 \\
                        \xi^{k_{1}} &    \xi^{k_{2}} & \cdots & \xi^{k_{\delta}} \\
                         \vdots &\vdots & \ddots & \vdots \\
                         \xi^{(\delta-1)k_{1}} &   \xi^{(\delta-1)k_{2}} & \cdots & \xi^{(\delta-1)k_{\delta}} \\
                         \end{array}
                       \right) \\
            &=& \prod_{1\leq i <j \leq \delta}(\xi^{k_{j}}-\xi^{k_{i}}) \neq 0,
         \end{eqnarray*}
which contradicts to the \textbf{Fact 1}.

 \textbf{Case (ii)}: $t=\delta-1$. At this time, by Eq. (4), we have
 \[\bm{u}_{\delta-1}= \rho(1, \xi^{(\delta-1)k_{1}},\ldots,\xi^{(\delta-1)k_{\delta-1}},0),\]
and
\[\bm{v}_{\delta-1}\triangleq \bm{u}_{\delta-1,0}-\frac{1}{\rho}\bm{u}_{\delta-1}=(0, 0, \ldots, 0, \xi^{(\delta-1)k_{\delta}}).\]
 Let  \[ M=\left(
           \begin{array}{c}
             A \\
             \bm{v}_{\delta-1} \\
           \end{array}
         \right) .\]
          Then
          \begin{eqnarray*}
          % \nonumber to remove numbering (before each equation)
            \det(M) &=& \det\left(
                         \begin{array}{ccccc}
                           1 & 1 & \cdots & 1 & 1 \\
                           1 & \xi^{k_{1}} & \cdots & \xi^{k_{\delta-1}}& \xi^{k_{\delta}} \\
                           \vdots & \vdots & \ddots & \vdots & \vdots \\
                           1 & \xi^{(\delta-1)k_{1}} & \cdots & \xi^{(\delta-1)k_{\delta-1}}& \xi^{(\delta-1)k_{\delta}} \\
                           0 & 0 & \cdots &0& \xi^{(\delta-1)k_{\delta}} \\
                         \end{array}
                       \right) \\
             &=& \xi^{(\delta-1)k_{\delta}}\det\left(
                         \begin{array}{cccc}
                          1 &   1 & \cdots & 1 \\
                        1 &    \xi^{k_{1}} & \cdots & \xi^{k_{\delta-1}} \\
                         \vdots &\vdots & \ddots & \vdots \\
                         1 &   \xi^{(\delta-1)k_{1}} & \cdots & \xi^{(\delta-1)k_{\delta-1}} \\
                         \end{array}
                       \right)\\
   &=& \xi^{(\delta-1)k_{\delta}} \prod_{0\leq i < j \leq \delta-1}(\xi^{k_{j}}-\xi^{k_{i}}) \neq 0,
          \end{eqnarray*}
 which also contradicts to the \textbf{Fact 1}.

 \textbf{Case (iii)}: $t=\delta$. Recall that $a\rho+b(r+\delta-1)=\delta$ , $\alpha=\xi^{\rho}$ and $\gamma=\alpha^{a}$. Since $\rho \mid k_{j}$, we have
\[\gamma^{k_{j}}=(\alpha^{a\rho})^{\frac{k_{j}}{\rho}}=(\alpha^{\delta-b(r+\delta-1)})^{\frac{k_{j}}{\rho}}
=\alpha^{\frac{k_{j}}{\rho}\delta}=(\xi^{\rho})^{\frac{k_{j}}{\rho}\delta}=\xi^{\delta k_{j}},\]
for $j=1, 2, \ldots, \delta.$
Let
\[M=\left(
      \begin{array}{c}
        A \\
        \bm{\gamma} \\
      \end{array}
    \right).
 \]
Then
\begin{eqnarray*}
% \nonumber to remove numbering (before each equation)
  \det(M) &=& \det\left(
      \begin{array}{ccccc}
        1 & 1 & 1 & \cdots & 1 \\
        1 & \xi^{k_{1}} & \xi^{k_{2}}  & \cdots & \xi^{k_{\delta}} \\
        1 & \xi^{2k_{1}} & \xi^{2k_{2}}  & \cdots & \xi^{2k_{\delta}} \\
        \vdots & \vdots & \vdots & \ddots & \vdots \\
        1 & \xi^{\delta k_{1}} & \xi^{\delta k_{2}}  & \cdots & \xi^{\delta k_{\delta}} \\
      \end{array}
    \right) \\
   &=& \prod_{0\leq i < j \leq \delta}(\xi^{k_{j}}-\xi^{k_{i}}) \neq 0,
\end{eqnarray*}
  which still contradicts to the \textbf{Fact 1}.

In each case, it always leads to a contradiction. Thus $d \geq \delta+2$. The proof is completed.
\end{proof}

\begin{remark}
\cite[Theorem 1 (2)]{3} can be seen as the $\delta=2$ case of Theorem 2.
\end{remark}

\begin{example}
Let $r=4$, $\delta=6$ and $q=19$,  then by Theorem 2, for any $n=27n'$ with $\gcd(n', 57)=1$, there exists a $19$-ary optimal cyclic $(4, 6)$-LRC of length $n$ and minimum distance $8$.
\end{example}

If we further assume that ``$r \geq \delta+1$ and $\gcd(\frac{n}{r+\delta -1}, r+\delta -1)=1$'' in Theorem 2, then we can obtain an optimal cyclic $(r, \delta)$-LRC with larger minimum distance as follows.
\begin{theorem}
Let $q$ be a prime power and $n$ be a positive integer with $\gcd(n, q)=1$. Let $3 \leq \delta+1 \leq r$ such that $(r+\delta-1) \mid \gcd(n, q-1)$ and $\gcd(\frac{n}{r+\delta-1}, r+\delta-1)=1$. Then there exists a $q$-ary optimal cyclic $(r, \delta)$-LRC of length $n$ and minimum distance $2\delta$.
\end{theorem}

\begin{proof}
  Let $\xi \in \mathbb{F}_{q^{s}}$ be
a primitive $n$-th root of unity and $\alpha = \xi^{\rho}$ be a prmitive $(r+\delta-1)$-th root of unity, where $\rho=\frac{n}{r+\delta-1}$.  Since $(r+\delta-1) \mid (q-1)$, $\alpha^{q-1}=(\xi^{\rho})^{q-1}=(\xi^{n})^{\frac{q-1}{r+\delta-1}}=1$, thus $\alpha \in \mathbb{F}_{q}$. Since $\gcd(\rho, r+\delta-1)=1$, there exist integers $a, b$, such that $a\rho+b(r+\delta-1)=1$. Let $\gamma = \alpha^{a} \in \mathbb{F}_{q}$. Then
\[\gamma^{\rho}=\alpha^{a\rho}=\alpha^{1-b(r+\delta-1)}=\alpha.\]

Let
\[g(x)=(x-1)\prod_{i=1}^{\delta-1}(x^{\rho}-\alpha^{i})\prod_{j=\delta}^{2\delta-2}(x-\gamma^{j}).\]
Note that $(\gamma^{j})^{\rho} = \alpha^{j} \neq \alpha^{i}$ for any $1 \leq i \leq \delta-1$ and $\delta \leq j \leq 2\delta-2$. Then $g(x)$ is a polynomial over $\mathbb{F}_{q}$ and $g(x) \mid (x^{n}-1)$ since all roots of $g(x)$ are $n$-th roots of unity and they are distinct. Let $C$ be the cyclic code with generator polynomial $g(x)$. Then the dimension of $C$ is $k=n-\deg(g(x))=n-((\delta-1)\rho+\delta)=r\rho-\delta$. Note that the set of the roots of $g(x)$ contains $\bigcup_{\ell=1}^{\delta-1}L_{\ell},$
where
 \[L_{\ell}=\{\xi^{i} \mid i\textnormal{ mod }(r+\delta-1) =\ell \}.\]
By Proposition 1, $C$ has $(r, \delta)$-locality.  Note that
$\lceil\frac{k}{r}\rceil=\lceil\frac{r\rho-\delta}{r}\rceil=\rho$ since $r \geq \delta+1$. By the bound
(2),
\[d \leq n-(r\rho-\delta)-(\rho-1)(\delta-1)+1=2\delta.\]
Thus to prove $C$ is optimal, we only need to prove that $d \geq 2\delta$. By contradiction, we suppose $d \leq 2\delta-1$. Then there exists a nonzero polynomial $c(x)=\sum_{i=0}^{2\delta-2}c_{i}x^{k_{i}}$ with $0=k_{0} < k_{1}< \cdots < k_{2\delta-2} <n$, such that $g(x) \mid c(x)$. Since $\xi$ is a primitive $n$-th root of unity, we have
\[\xi^{i} \neq \xi^{j},\textnormal{ for any }0 \leq i \neq j \leq 2\delta-2 .\]
Set
\[\bm{1}=(1,1,\ldots,1),\]
\[\bm{\gamma}_{\ell}=\left(1, \gamma^{\ell k_{1}},\ldots, \gamma^{\ell k_{2\delta-2}}\right), \ell=\delta, \delta+1, \ldots, 2\delta-2,\]
and
\[\bm{u}_{i,j}=\left(1,\xi^{(i+j(r+\delta-1))k_{1}},\ldots ,\xi^{(i+j(r+\delta-1))k_{2\delta-2}}\right), i=1, 2, \ldots, \delta-1, j=0, 1, \ldots, \rho-1,\]
be the vectors in $\mathbb{F}^{2\delta-1}_{q^{s}}$.
Let $U$ be the vector space spanned by the vectors $\bm{1}, \bm{\gamma}_{\ell}$ $(\ell=\delta, \delta+1, \ldots, 2\delta-2)$ and $\bm{u}_{i,j}$ ($i=1, 2, \ldots, \delta-1, j=0, 1, \ldots, \rho-1$) over $\mathbb{F}_{q^{s}}$. Similar to the proof of Theorem 2, we have
\begin{center}
\textbf{Fact 2}: Let $M$ be a $(2\delta-1)\times (2\delta-1)$ matrix whose row vectors belong to $U$, then $\det(M)=0$.
\end{center}

Now, assume that among $2\delta-2$ integers $k_{i}$, $t$ integers are divisible by $\rho$ and the rest $2\delta-2-t$ integers are not divisible by $\rho$. Then $0 \leq t \leq 2\delta-2$. Without loss of generality, we suppose that $\rho \mid k_{1}, \ldots, k_{t}$ and $\rho \nmid k_{t+1}, \ldots, k_{2\delta-2}$. Then for any $1 \leq i \leq \delta-1$,  we have
\begin{equation}\label{5}
 \sum_{j=0}^{\rho-1}\xi^{(i+j(r+\delta-1))k_{\ell}}=\left\{
\begin{aligned}
   \rho\xi^{ik_{\ell}}, & \textnormal{ if }1 \leq \ell \leq t, \\
 0, & \textnormal{ if } t+1 \leq j \leq 2\delta-2.
\end{aligned}\right.
\end{equation}
Let
\[A=\left(
      \begin{array}{c}
        \bm{1} \\
        \bm{u}_{1,0} \\
        \vdots \\
        \bm{u}_{\delta-1,0}  \\
      \end{array}
    \right)
\]
be a $\delta \times (2\delta-1)$ matrix over $\mathbb{F}_{q^{s}}$.

\textbf{Case (i)}:  $\delta-1 \leq t \leq 2\delta-2$. For $i=1, 2, \ldots, \delta-1$, by Eq. (5), we have
\[\bm{u}_{i} \triangleq \bm{u}_{i,0}+\bm{u}_{i,1}+ \cdots +\bm{u}_{i,\rho-1}=\rho(1, \xi^{ik_{1}},\ldots, \xi^{ik_{t}}, 0, \ldots, 0), \]
and
\[\bm{v}_{i} \triangleq \bm{u}_{i,0}-\frac{1}{\rho}\bm{u}_{i}=(0, \ldots, 0, \xi^{ik_{t+1}},\ldots, \xi^{ik_{2\delta-2}}). \]
Thus $\bm{v}_{i} \in U.$ Recall that $a\rho+b(r+\delta-1)=1$, $\alpha=\xi^{\rho}$ and $\gamma=\alpha^{a}$, we have
\[\gamma^{k_{j}}=(\alpha^{a\rho})^{\frac{k_{j}}{\rho}}=(\alpha^{1-b(r+\delta-1)})^{\frac{k_{j}}{\rho}}
=\alpha^{\frac{k_{j}}{\rho}}=(\xi^{\rho})^{\frac{k_{j}}{\rho}}=\xi^{k_{j}},\]
for $j=1, 2,\ldots, t$.
Let
    \[ M=\left(
           \begin{array}{c}
             A \\
             \bm{v}_{1} \\
             \vdots \\
             \bm{v}_{2\delta-2-t}  \\
            \bm{\gamma}_{\delta} \\
            \vdots \\
            \bm{\gamma}_{t} \\
           \end{array}
         \right).\]
Then $M$ is a $(2\delta-1)\times (2\delta-1)$ matrix whose row vectors belong to $U$, and
          \begin{eqnarray*}
          % \nonumber to remove numbering (before each equation)
            \det(M) &=& \det\left(
                         \begin{array}{ccccccc}
                           1 & 1 & \cdots & 1 & 1 & \cdots & 1\\
                           1 & \xi^{k_{1}} & \cdots &\xi^{k_{t}} & \xi^{k_{t+1}} & \cdots & \xi^{k_{2\delta-2}}\\
                           \vdots & \vdots &\ddots &\vdots & \vdots  & \ddots & \vdots  \\
                           1 & \xi^{(\delta-1)k_{1}} & \cdots &\xi^{(\delta-1)k_{t}} & \xi^{(\delta-1)k_{t+1}} & \cdots & \xi^{(\delta-1)k_{2\delta-2}}\\
                           0 & 0 & \cdots &0 & \xi^{k_{t+1}} & \cdots & \xi^{k_{2\delta-2}}\\
                           \vdots & \vdots &\ddots &\vdots & \vdots  & \ddots & \vdots  \\
                           0 & 0 & \cdots &0 & \xi^{(2\delta-2-t)k_{t+1}} & \cdots & \xi^{(2\delta-2-t)k_{2\delta-2}}\\
                    1 & \xi^{\delta k_{1}} & \cdots &\xi^{\delta k_{t}} & \gamma^{\delta k_{t+1}} & \cdots & \gamma^{\delta k_{2\delta-2}}\\
                           \vdots & \vdots &\ddots &\vdots & \vdots  & \ddots & \vdots  \\
                   1 & \xi^{t k_{1}} & \cdots &\xi^{t k_{t}} & \gamma^{t k_{t+1}} & \cdots & \gamma^{t k_{2\delta-2}}\\
                         \end{array}
                       \right) \\
             &=& ab,
          \end{eqnarray*}
where
\begin{eqnarray*}
% \nonumber to remove numbering (before each equation)
  a &=& \det\left(
                                     \begin{array}{cccc}
                                       1 & 1 & \cdots &1 \\
                                       1 & \xi^{k_{1}} & \cdots &\xi^{k_{t}} \\
                                       \vdots & \vdots &\ddots &\vdots \\
                                       1 & \xi^{t k_{1}} & \cdots &\xi^{t k_{t}} \\
                                     \end{array}
                                   \right) \\
    &=& \prod_{0 \leq i < j \leq t}(\xi^{k_{j}}-\xi^{k_{i}})\neq 0,
\end{eqnarray*}
and
\begin{eqnarray*}
% \nonumber to remove numbering (before each equation)
  b &=& \det\left(
                                     \begin{array}{cccc}
                                      \xi^{k_{t+1}} & \xi^{k_{t+2}} & \cdots & \xi^{k_{2\delta-2}} \\
                                       \xi^{2k_{t+1}} & \xi^{2k_{t+2}} & \cdots & \xi^{2k_{2\delta-2}} \\
                                       \vdots & \vdots &\ddots &\vdots \\
                                       \xi^{(2\delta-2-t)k_{t+1}} & \xi^{(2\delta-2-t)k_{t+2}} & \cdots & \xi^{(2\delta-2-t)k_{2\delta-2}} \\
                                     \end{array}
                                   \right) \\
    &=& \xi^{k_{t+1}+\cdots+k_{2\delta-2}}\prod_{t+1 \leq i < j \leq 2\delta-2}(\xi^{k_{j}}-\xi^{k_{i}}) \neq 0.
\end{eqnarray*}
Thus $\det(M) \neq 0$ which  contradicts to the \textbf{Fact 2}.

 \textbf{Case (ii)}: $0 \leq t \leq \delta-2$. By Eq. (5), for $i=1, 2,\ldots, \delta-1$, we have
\[\bm{u}_{i} =\bm{u}_{i,0}+\bm{u}_{i,1}+ \cdots +\bm{u}_{i,\rho-1}=\rho(1, \xi^{ik_{1}},\ldots, \xi^{ik_{t}}, 0, \ldots, 0), \]
Since $t \leq \delta-2$, $t+1 \leq \delta-1$. We may let
\[B=\frac{1}{\rho}\left(
      \begin{array}{c}
        \bm{u}_{1}\\
        \bm{u}_{2} \\
         \vdots \\
        \bm{u}_{t+1}\\
      \end{array}
    \right)=\left(
              \begin{array}{ccccccc}
                1 & \xi^{k_{1}} & \ldots & \xi^{k_{t}} & 0 & \ldots & 0 \\
                1 & \xi^{2k_{1}} & \ldots & \xi^{2k_{t}} & 0 & \ldots & 0 \\
                \vdots & \vdots & \ddots & \vdots & \vdots & \ddots & \vdots \\
                1 & \xi^{(t+1)k_{1}} & \ldots & \xi^{(t+1)k_{t}} & 0 & \ldots & 0 \\
              \end{array}
            \right)
    .\]
The first $t+1$ columns of $B$ form a Vandermonde matrix which is invertible. Let
\[\bm{e}_{i}=(\underbrace{0,\ldots, 0}_{(i-1) \textnormal{ times}}, 1, 0, \ldots, 0).\]
Then we deduce that for any $1 \leq i \leq t+1$,
\[\bm{e}_{i}  \in U.\]
 Let
\[m = \max_{t+1 \leq i \leq 2\delta-2}|S_{i}|, \]
where
\[S_{i} \triangleq \{t+1 \leq j \leq 2\delta-2 : j \neq i\textnormal{ and }k_{j}-k_{i} \equiv 0 \textnormal{ (mod }\rho) \}.\]
Then $0 \leq m \leq 2\delta-t-3$.
\begin{description}
  \item[1)] if $\delta-1 \leq m \leq 2\delta-t-3$ :  Without loss of generality, we suppose that $\rho \mid (k_{2\delta-2-m}-k_{2\delta-2}), (k_{2\delta-1-m}-k_{2\delta-2}),\ldots, (k_{2\delta-3}-k_{2\delta-2})$. Set $h=2\delta-2-m$, then $h \leq \delta-1$.
For any $1 \leq i \leq \delta-1$,  we have
\begin{equation}\label{6}
 \sum_{j=0}^{\rho-1}\xi^{(i+j(r+\delta-1))(k_{\ell}-k_{2\delta-2})}=\left\{
\begin{aligned}
   \rho\xi^{i(k_{\ell}-k_{2\delta-2})}, & \textnormal{ if }h \leq \ell \leq 2\delta-2, \\
 0, & \textnormal{ if } 1 \leq \ell \leq h-1.
\end{aligned}\right.
\end{equation}
Thus
\begin{eqnarray*}
% \nonumber to remove numbering (before each equation)
  \bm{u}'_{i} &\triangleq& \xi^{-ik_{2\delta-2}}\bm{u}_{i,0}+\xi^{-(i+r+\delta-1)k_{2\delta-2}}\bm{u}_{i,1}+ \cdots + \xi^{-(i+(\rho-1)(r+\delta-1))k_{2\delta-2}}\bm{u}_{i,\rho-1} \\
   &=& (0, \ldots, 0, \rho\xi^{i(k_{h}-k_{2\delta-2})}, \ldots, \rho \xi^{i(k_{2\delta-3}-k_{2\delta-2})}, \rho),
\end{eqnarray*}
and
\[\bm{v}'_{i} \triangleq \bm{u}_{i,0}-\frac{1}{\rho}\xi^{i k_{2\delta-2}}\bm{u}'_{i}=(1, \xi^{i k_{1}}, \xi^{i k_{2}},\ldots, \xi^{i k_{h-1}}, 0\ldots, 0).\]
Thus $\bm{v}'_{i} \in U.$ On the other hand, for $h \leq \ell \leq 2\delta-2$, since $\rho \mid (k_{\ell}-k_{2\delta-2})$, we have
\begin{eqnarray*}
% \nonumber to remove numbering (before each equation)
  \gamma^{k_{\ell}-k_{2\delta-2}} &=& (\alpha^{a\rho})^{\frac{k_{j}-k_{2\delta-2}}{\rho}}=(\alpha^{1-b(r+\delta-1)})^{\frac{k_{\ell}-k_{2\delta-2}}{\rho}} \\
   &=& \alpha^{\frac{k_{\ell}-k_{2\delta-2}}{\rho}}=(\xi^{\rho})^{\frac{k_{\ell}-k_{2\delta-2}}{\rho}}=\xi^{k_{\ell}-k_{2\delta-2}},
\end{eqnarray*}
i.e.,
\[\gamma^{k_{\ell}}=\epsilon \xi^{k_{\ell}},\]
where $\epsilon=(\frac{\gamma}{\xi})^{k_{2\delta-2}}$.
Let
    \[ M=\left(
           \begin{array}{c}
             A \\
             \bm{v}'_{1} \\
             \vdots \\
             \bm{v}'_{h}  \\
            \frac{1}{\epsilon}\bm{\gamma}_{\delta} \\
            \vdots \\
            \frac{1}{\epsilon}\bm{\gamma}_{m} \\
           \end{array}
         \right).\]
Then $M$ is a $(2\delta-1)\times (2\delta-1)$ matrix whose row vectors belong to $U$, and
          \begin{eqnarray*}
          % \nonumber to remove numbering (before each equation)
            \det(M) &=& \det\left(
                         \begin{array}{ccccccc}
                           1 & 1 & \cdots & 1 & 1 & \cdots & 1\\
                           1 & \xi^{k_{1}} & \cdots &\xi^{k_{h-1}} & \xi^{k_{h}} & \cdots & \xi^{k_{2\delta-2}}\\
                           \vdots & \vdots &\ddots &\vdots & \vdots  & \ddots & \vdots  \\
                           1 & \xi^{(\delta-1)k_{1}} & \cdots &\xi^{(\delta-1)k_{h-1}} & \xi^{(\delta-1)k_{h}} & \cdots & \xi^{(\delta-1)k_{2\delta-2}}\\
                          1 & \xi^{k_{1}} & \cdots &\xi^{k_{h-1}} & 0 & \cdots & 0\\
                           \vdots & \vdots &\ddots &\vdots & \vdots  & \ddots & \vdots  \\
                          1 & \xi^{h k_{1}} & \cdots &\xi^{h k_{h-1}} & 0 & \cdots & 0\\
                    1 & \frac{1}{\epsilon}\gamma^{\delta k_{1}} & \cdots &\frac{1}{\epsilon}\gamma^{\delta k_{h-1}} & \xi^{\delta k_{h}} & \cdots & \xi^{\delta k_{2\delta-2}}\\
                           \vdots & \vdots &\ddots &\vdots & \vdots  & \ddots & \vdots  \\
                   1 & \frac{1}{\epsilon}\gamma^{m k_{1}} & \cdots &\frac{1}{\epsilon}\gamma^{m k_{h-1}} & \xi^{m k_{h}} & \cdots & \xi^{m k_{2\delta-2}}\\
                         \end{array}
                       \right) \\
             &=& ab,
          \end{eqnarray*}
where
\begin{eqnarray*}
% \nonumber to remove numbering (before each equation)
  a &=& \det\left(
                                     \begin{array}{cccc}
                                       1 & \xi^{k_{1}} & \cdots &\xi^{k_{h-1}} \\
                                       1 & \xi^{2k_{1}} & \cdots &\xi^{2k_{h-1}} \\
                                       \vdots & \vdots &\ddots &\vdots \\
                                       1 & \xi^{h k_{1}} & \cdots &\xi^{h k_{h-1}} \\
                                     \end{array}
                                   \right) \\
    &=& \prod_{1 \leq i < j \leq h-1}(\xi^{k_{j}}-\xi^{k_{i}})\neq 0,
\end{eqnarray*}
and
\begin{eqnarray*}
% \nonumber to remove numbering (before each equation)
  b &=& \det\left(
                                     \begin{array}{cccc}
                                       1 & 1 & \cdots & 1 \\
                                       \xi^{k_{h}} & \xi^{k_{h+1}} & \cdots & \xi^{k_{2\delta-2}} \\
                                       \vdots & \vdots &\ddots &\vdots \\
                                       \xi^{m k_{h}} & \xi^{m k_{h+1}} & \cdots & \xi^{m k_{2\delta-2}} \\
                                     \end{array}
                                   \right) \\
    &=& \prod_{h \leq i < j \leq 2\delta-2}(\xi^{k_{j}}-\xi^{k_{i}}) \neq 0.
\end{eqnarray*}
Thus $\det(M) \neq 0$ which contradicts to the \textbf{Fact 2}.
   \item[2)] If $0 \leq m \leq \delta-2$:  Note that we have proved that $\bm{e}_{i} \in U$, for $1 \leq i \leq t+1$. We claim that for  each $i$ with $t+2 \leq i \leq 2\delta-1$, we still have
\[\bm{e}_{i} \in U.\]
Without loss of generality, we only prove it for $i=t+2$. We consider the set $S_{t+1}$.  Without loss of generality, we suppose that $S_{t+1}=\{k_{t+2}, k_{t+3}, \ldots, k_{t+m'+1}\}$, i.e., $\rho \mid (k_{t+2}-k_{t+1}), (k_{t+3}-k_{t+1}),\ldots, (k_{t+m'+1}-k_{t+1})$, where $m'=|S_{t+1}|$ ($m'=0$ if $S_{t+1}=\emptyset$). Similarly as Eq. (6), for any $1 \leq i \leq \delta-1$,  we have
\begin{equation*}
  \sum_{j=0}^{\rho-1}\xi^{(i+j(r+\delta-1))(k_{\ell}-k_{t+1})}=\left\{
\begin{aligned}
   \rho\xi^{i(k_{\ell}-k_{t+1})}, & \textnormal{ if }t+1 \leq \ell \leq t+m'+1, \\
 0, &~\ell \textnormal{ otherwise}.
\end{aligned}\right.
\end{equation*}
Thus for any $1 \leq i \leq \delta-1$,
\begin{eqnarray*}
% \nonumber to remove numbering (before each equation)
  \bm{w}_{i} &\triangleq& \xi^{-ik_{t+1}}\bm{u}_{i,0}+\xi^{-(i+r+\delta-1)k_{t+1}}\bm{u}_{i,1}+ \cdots + \xi^{-(i+(\rho-1)(r+\delta-1))k_{t+1}}\bm{u}_{i,\rho-1} \\
   &=& (\underbrace{0, \ldots, 0}_{(t+1) \textnormal{ times}}, \rho, \rho\xi^{i(k_{t+2}-k_{t+1})}, \ldots, \rho \xi^{i(k_{t+m'+1}-k_{t+1})},\bm{0}),
\end{eqnarray*}
and $\bm{w}_{i} \in U$. Since $m' \leq m \leq \delta-2$, we may let
\[ B=\frac{1}{\rho}\left(
           \begin{array}{c}
             \bm{w}_{1} \\
             \bm{w}_{2} \\
             \vdots\\
             \bm{w}_{m'+1} \\
           \end{array}
         \right)\]
be a matrix whose rows are belong to $U$. Then the $(t+2)$-th column to the $(t+m'+2)$-th column of $B$ form an $(m'+1) \times (m'+1)$ Vandermonde  matrix which is invertible. It deduces that $\bm{e}_{t+2} \in U$. The claim is proved. At this time, the $2\delta-1$ vectors $\bm{e}_{i}$ all in $U$, i.e., $\dim(U)=2\delta-1$ which is a contradiction.
\end{description}

In each case, it always leads to a contradiction. Thus $d \geq 2\delta$. The proof is completed.
\end{proof}
\begin{remark}
 For any $\delta+1 \leq d \leq 2\delta, r \geq d-\delta+1$ and with other conditions of Theorem 3, let
\[g(x)=(x-1)\prod_{i=1}^{\delta-1}(x^{\rho}-\alpha^{i})\prod_{j=\delta}^{d-2}(x-\gamma^{j}).\]
We can prove similarly that the cyclic code $C$ generated by $g(x)$ is a $q$-ary optimal $(r, \delta)$-LRC of length $n$ and minimum distance $d$.
\end{remark}
\begin{example}
Let $r=4$, $\delta=3$ and $q=7$,  then by Theorem 3, for any $n=6n'$ with $\gcd(n', 42)=1$, there exists a $7$-ary optimal cyclic $(4, 3)$-LRC of length $n$ and minimum distance $6$.
\end{example}

When $\delta=3$, we provide another construction of optimal $(r, 3)$-LRCs with unbounded length in the following theorem, which is just a modification of Theorem 3.

\begin{theorem}
Let $q$ be a prime power and $n$ be an odd integer with $\gcd(n, q)=1$. Let $r \geq 4$ such that $(r+2) \mid \gcd(n, q+1)$ and $\gcd(\frac{n}{r+2}, r+2)=1$. Then there exists a $q$-ary optimal cyclic $(r, 3)$-LRC of length $n$ and minimum distance 6.
\end{theorem}

\begin{proof}
  Let $\xi \in \mathbb{F}_{q^{s}}$ be
a primitive $n$-th root of unity and $\alpha = \xi^{\rho}$ be a primitive $(r+2)$-th root of unity, where $\rho=\frac{n}{r+2}$. Since $(r+2) \mid (q+1)$, $\alpha^{q+1}=(\xi^{\rho})^{q+1}=(\xi^{n})^{\frac{q+1}{r+2}}=1$, i.e., $\alpha^{q}=\alpha^{-1}$, thus $\alpha \in \mathbb{F}_{q^{2}}$. Since $\gcd(\rho, r+2)=1$, there exist integers $a, b$, such that $a\rho+b(r+2)=1$. Let $\gamma = \alpha^{a} \in \mathbb{F}_{q^{2}}$. Then $\gamma^{q}=\alpha^{aq}=\alpha^{-a}=\gamma^{-1}$ and
\[\gamma^{\rho}=\alpha^{a\rho}=\alpha^{1-b(r+2)}=\alpha.\]

Let
\[g(x)=(x-1)(x^{\rho}-\alpha)(x^{\rho}-\alpha^{-1})(x-\gamma^{2})(x-\gamma^{-2}).\]
Then $g(x) \mid (x^{n}-1)$ since all roots of $g(x)$ are $n$-th roots of unity and they are distinct. Put $g(x)=\sum_{i=0}^{m}g_{i}x^{i}$, then
\begin{eqnarray*}
% \nonumber to remove numbering (before each equation)
  \sum_{i=0}^{m}g^{q}_{i}x^{i} &=& (x-1)(x^{\rho}-\alpha^{q})(x^{\rho}-\alpha^{-q})(x-\gamma^{2q})(x-\gamma^{-2q}) \\
   &=& (x-1)(x^{\rho}-\alpha^{-1})(x^{\rho}-\alpha)(x-\gamma^{-2})(x-\gamma^{2}) \\
    &=& g(x).
\end{eqnarray*}
Thus $g_{i}^{q}=g_{i}$, $g(x)$ is a polynomial over $\mathbb{F}_{q}$. Let $C$ be the cyclic code with generator polynomial $g(x)$. Then the dimension of $C$ is $k=n-\deg(g(x))=n-(2\rho+3)=r\rho-3$. Note that the set of the roots of $g(x)$ contains $L_{1} \bigcup L_{-1},$
where
 \[L_{\ell}=\{\xi^{i} \mid i\textnormal{ mod }(r+2) =\ell \}, \ell=1, -1.\]
Since $n$ is odd, $C$ has $(r, 3)$-locality from Proposition 1.  Note that
$\lceil\frac{k}{r}\rceil=\lceil\frac{r\rho-3}{r}\rceil=\rho$ since $r \geq 4$. By the bound
(2),
\[d \leq n-(r\rho-3)-2(\rho-1)+1=6.\]
The theorem then follows from the following claim:
\[ \textbf{Claim}: d \geq 6.\]
The method of the proof of this claim is completely similar to Theorem 3. So we leave the proof in Appendix.

\end{proof}

\begin{example}
Let $r=4$, $\delta=3$ and $q=5$,  then by Theorem 4, for any $n=6n'$ with $\gcd(n', 30)=1$, there exists a $5$-ary optimal cyclic $(4, 3)$-LRC of length $n$ and minimum distance $6$.
\end{example}

\section{Conclusion}

In this paper,  we  construct several families of optimal $(r, \delta)$-LRCs via cyclic codes. In particular, for any $\delta +1 \leq d \leq 2\delta$, there always exists a $q$-ary optimal cyclic $(r, \delta)$-LRC with minimum distance $d$ and  unbounded length, that is the length of the code is independent of the alphabet size $q$. Recently, when the minimum distance $d \geq 5$, V. Guruswami et al. \cite{17} proved that the code length $n$ of a $q$-ary optimal $r$-LRC is upper bounded by $\mathcal{O}(dq^{3})$ (roughly). Thus, it is interesting to study the upper bound of the length of a  $q$-ary optimal $(r, \delta)$-LRC with minimum distance $d \geq 2\delta + 1$ in the future work.

\vskip 3mm \noindent {\bf Acknowledgments}  This research is supported by the 973 Program of China (Grant No. 2013CB834204), the National Natural Science Foundation of China (Grant No. 61571243), and the Fundamental Research Funds for the Central Universities of China.

 \section*{Appendix: Proof of the Claim in Theorem 4}
By contradiction, we suppose $d \leq 5$. Then there exists a nonzero polynomial $c(x)=\sum_{j=0}^{4}c_{j}x^{k_{j}}$ with $c_{0} \neq 0$ and $0 < k_{1}< \cdots < k_{4} <n$, such that $g(x) \mid c(x)$.
Let
\[\bm{1}=(1,1,1,1,1),\]
\[\bm{\gamma}_{1}=(1, \gamma^{2k_{1}},\gamma^{2k_{2}},\gamma^{2k_{3}},\gamma^{2k_{4}}),\]
\[\bm{\gamma}_{2}=(1, \gamma^{-2k_{1}},\gamma^{-2k_{2}},\gamma^{-2k_{3}},\gamma^{-2k_{4}}),\]
\[\bm{u}_{i}=\left(1,\xi^{(1+i(r+2))k_{1}},\xi^{(1+i(r+2))k_{2}}, \xi^{(1+i(r+2))k_{3}} ,\xi^{(1+i(r+2))k_{4}}\right),\]
and
\[\bm{v}_{i}=\left(1,\xi^{(-1+i(r+2))k_{1}},\xi^{(-1+i(r+2))k_{2}}, \xi^{(-1+i(r+2))k_{3}}, \xi^{(-1+i(r+2))k_{4}}\right),i=0, 1, \ldots, \rho-1,\]
be the vectors in $\mathbb{F}^{5}_{q^{s}}$.
Let $U$ be the vector space spanned by the vectors $\bm{1}, \bm{\gamma}_{1}, \bm{\gamma}_{2}, \bm{u}_{i}$ and $\bm{v}_{i}$ ($i=0, 1, \ldots, \rho-1$) over $\mathbb{F}_{q^{s}}$. Similar to Theorem 2, we have
\begin{center}
\textbf{Fact 3}: Let $M$ be a $5\times 5$ matrix whose row vectors belong to $U$, then $\det(M)=0$.
\end{center}

Now, assume that among 4 integers $k_{i}$, $t$ integers are divisible by $\rho$ and the rest $4-t$ integers are not divisible by $\rho$. Then $0 \leq t \leq 4$. Without loss of generality, we suppose that $\rho \mid k_{1}, \ldots, k_{t}$ and $\rho \nmid k_{t+1}, \ldots, k_{4}$.

 \textbf{Case (i)}: $t=0$. Then for any $1\leq j \leq 4$,
 \[\sum_{i=0}^{\rho-1}(\xi^{1+i(r+2)})^{k_{j}}=0.\]
 Note that for any $0 \leq i \leq \rho-1$, $c(\xi^{1+i(r+2)})=0$. Thus
 \[0=\sum_{i=0}^{\rho-1}c(\xi^{1+i(r+2)})=\rho c_{0},\]
 which leads to $c_{0} = 0$, contradiction.

\textbf{Case (ii)}: $1 \leq t \leq 2$. Let
\[\bm{u} \triangleq \bm{u}_{0}+\bm{u}_{1}+ \cdots +\bm{u}_{\rho-1}, \]

\[\bm{v} \triangleq \bm{v}_{0}+\bm{v}_{1}+ \cdots +\bm{v}_{\rho-1} \]
and
    \[ M=\left(
           \begin{array}{c}
             \bm{1} \\
             \bm{u_{0}} \\
             \bm{v_{0}} \\
             \frac{1}{\rho}\bm{u} \\
             \frac{1}{\rho}\bm{v} \\
           \end{array}
         \right) .\]

 \textbf{Case (iii)}: $t=3$. Let
\[\bm{v}' \triangleq \bm{v}_{0}-\frac{1}{\rho}\bm{v}=(0, 0, 0, 0, \xi^{-k_{4}}). \]
Recall that
\[a\rho+b(r+2)=1, \alpha=\xi^{\rho}\textnormal{ and }\gamma=\alpha^{a},\]
we have
\[\gamma^{k_{i}}=(\alpha^{a\rho})^{\frac{k_{i}}{\rho}}=(\alpha^{1-b(r+2)})^{\frac{k_{i}}{\rho}}
=\alpha^{\frac{k_{i}}{\rho}}=(\xi^{\rho})^{\frac{k_{i}}{\rho}}=\xi^{k_{i}}, i=1,2,3.\]
Let
    \[ M=\left(
           \begin{array}{c}
             \bm{1} \\
             \bm{u}_{0} \\
             \bm{v}_{0} \\
             \bm{\gamma}_{1} \\
             \bm{v}' \\
           \end{array}
         \right) .\]

\textbf{Case (iv)}: $t=4$. Similar to Case (iii), for $1 \leq i \leq 4$, we have
\[\gamma^{k_{i}}=\xi^{k_{i}}.\]
Let
    \[ M=\left(
           \begin{array}{c}
             \bm{1} \\
             \bm{u}_{0} \\
             \bm{v}_{0} \\
             \bm{\gamma}_{1} \\
            \bm{\gamma}_{2} \\
           \end{array}
         \right) .\]

In cases (ii)-(iv), it is not hard to verify that $\det(M) \neq 0$, which contradicts to the \textbf{Fact 3}.  Thus $d \geq 6$. The claim is proved.
\end{document}